\def\twocolumnmode{1}
\if\twocolumnmode0\def\breakifsinglecolumn{\if\twocolumnmode0\nonumber\\&\qquad{}\fi}\else\def\breakifsinglecolumn{}\fi

\if\twocolumnmode0
\documentclass[12pt,draftclsnofoot,onecolumn]{IEEEtran}
\else
\documentclass[10pt,conference]{IEEEtran}
\fi

\usepackage[utf8]{inputenc}
\usepackage[table,dvipsnames]{xcolor}
\usepackage{graphicx}
\usepackage{amsmath,amssymb}
\usepackage{xfrac}
\usepackage{pgffor}
\usepackage{tikz}
\usepackage{relsize}
\usepackage{tikz}
\usepackage{optidef}
\usepackage{pgfplots}
\usepackage{tikzscale}
\usepackage{amsthm}
\usepackage{algorithm}
\usepackage{algorithmic}
\usepackage{datetime}
\usepackage{microtype}
\usepackage{subcaption}
\usepackage[font=small,labelfont=bf]{caption}

\usetikzlibrary{spy}
\usetikzlibrary{arrows}
\usetikzlibrary{patterns}
\usetikzlibrary{plotmarks}
\usetikzlibrary{arrows.meta}
\usetikzlibrary{backgrounds}
\usetikzlibrary{positioning}
\usepgfplotslibrary{patchplots}

\theoremstyle{general} \newtheorem{theorem}{Theorem}
\theoremstyle{general} 
\theoremstyle{general} 
\theoremstyle{general} 
\theoremstyle{general} \newtheorem{p-corollary}{Corollary}[proposition]
\theoremstyle{general} 
\theoremstyle{general} 
\theoremstyle{remark}  \newtheorem{remark}{Remark}

\title{Precoding and Decoding Schemes for Downlink MIMO-RSMA with Simultaneous Diagonalization and User Exclusion}
\author{\IEEEauthorblockN{Rouaa Diab\rlap{\textsuperscript{\IEEEauthorrefmark{2}}}, Aravindh Krishnamoorthy\rlap{\textsuperscript{\IEEEauthorrefmark{2}\IEEEauthorrefmark{1}}},\,\,  and Robert Schober\rlap{\textsuperscript{\IEEEauthorrefmark{2}}}\\ \IEEEauthorblockA{\small \IEEEauthorrefmark{2}Friedrich-Alexander-Universit\"{a}t Erlangen-N\"{u}rnberg, Germany\\
\IEEEauthorrefmark{1}Fraunhofer Institute for Integrated Circuits (IIS) Erlangen, Germany}}\thanks{The authors acknowledge the financial support by the Federal Ministry of Education and Research of Germany in the programme of ``Souverän. Digital. Vernetzt.'' joint project 6G-RIC, project identification number: PIN 16KISK023.}\vspace{-0.5cm}}

\tikzset{new spy style/.style={spy scope={%
			magnification=2,
			size=0.5cm,
			connect spies,
			every spy on node/.style={
				rectangle,
				draw,
			},
			every spy in node/.style={
				draw,
				rectangle,
			}
		}
	}
} 

\IEEEoverridecommandlockouts

\begin{document}
\maketitle

\begin{abstract}
In this paper, we consider the precoder design for downlink multiple-input multiple-output (MIMO) rate-splitting multiple access (RSMA) systems. The proposed scheme with simultaneous diagonalization (SD) decomposes the MIMO channel matrices of the users into scalar channels via higher-order generalized singular value decomposition for the common message (CM) and block diagonalization (BD) for the private messages, thereby enabling low-complexity element-by-element successive interference cancellation (SIC) and decoding at the receivers. Furthermore, the proposed SD MIMO-RSMA overcomes a critical limitation in RSMA systems, whereby the achievable rate of the CM is restricted by the users with weak effective MIMO channel for the CM, by excluding a subset of users from decoding the CM. We formulate a non-convex weighted sum rate (WSR) optimization problem for SD MIMO-RSMA and solve it via successive convex approximation to obtain a locally optimal solution. Our simulation results reveal that, for both perfect and imperfect CSI, the proposed SD MIMO-RSMA with user exclusion outperforms baseline MIMO-RSMA schemes and linear BD precoding.
\end{abstract}

\section{Introduction}
The next generation multiple access (NGMA) schemes for 6th generation (6G) and beyond downlink communication systems must accommodate a large number of receivers and achieve high spectral and energy efficiencies, reliability, and flexibility in terms of resource allocation \cite{Liu2021}. However, conventional orthogonal multiple access (OMA) schemes have limited performance and can support only a small number of users as orthogonal resources are allocated to each user in the system. Hence, in recent years, the research on non-orthogonal schemes for future downlink wireless systems has gained momentum.

Non-orthogonal multiple access (NOMA) \cite{Liu2017}, based on superposition coding (SC) at the transmitter and successive-interference cancellation (SIC) at the receivers, has been shown to achieve high performance and user fairness \cite{Liu2017}. However, for NOMA, the decoding complexity increases significantly with the number of users. In fact, in a system with $K$ users, $\mathcal{O}(K^2)$ SIC stages are needed in total. 

On the other hand, rate-splitting multiple access (RSMA) \cite{Mao2018} is a non-orthogonal multiple access (MA) technique that generalizes orthogonal and non-orthogonal MA schemes \cite{Mao2018}, \cite{Clerckx2020}. In RSMA, each user's bit stream is split into a common message (CM), which is decoded by all users, and private messages (PMs), which are decoded only by the intended user. Furthermore, linear precoders are utilized at the transmitter and single-stage SIC at the receivers for decoding the CM and the PMs. The primary advantages of RSMA are its ability to manage interference in a robust manner and its improved performance for imperfect channel state information (CSI) at the base station (BS) \cite{Joudeh2016}, \cite{Mao2020}. This flexibility and robustness allows RSMA to achieve an enhanced performance over other MA schemes in many practical scenarios.

However, in multiple-input multiple-output (MIMO) RSMA systems, which can take advantage of the spatial degrees of freedom (DoFs) offered by MIMO channels to enhance performance, SIC and joint decoding of the elements of the MIMO symbol vector at the receivers entail a high complexity, limiting the practical deployment of MIMO-RSMA. Nevertheless, SIC and decoding complexity at the receivers can be substantially decreased by element-by-element decoding, which can be enabled by decomposing the effective MIMO-RSMA channels of the users into parallel scalar channels. While simultaneous diagonalization (SD) of the MIMO-RSMA channels for the PMs can be performed via conventional schemes such as block diagonalization (BD) \cite{Spencer2004} and zero forcing (ZF) \cite{Wiesel2008} precoding, corresponding techniques for SD of all MIMO user channels, as is needed for the CM, have not been proposed, yet. In \cite{Flores2019}, a MIMO-RSMA scheme combining a scalar CM, which is trivially simultaneously diagonalized, and BD precoding has been reported. However, for the MIMO-RSMA scheme in \cite{Flores2019} the CM cannot take advantage of the spatial DoFs offered by MIMO channels.

Furthermore, in RSMA, as the CM must be decoded by all users, the achievable rate of the CM is limited by the users with weak effective MIMO channels. Hence, developing precoding and decoding schemes that can overcome this limitation is of high interest.

To this end, in this paper, we propose novel SD MIMO-RSMA precoding and decoding schemes which can simultaneously diagonalize the effective MIMO channels for both the CM and the PMs of the users, thereby enabling element-by-element SIC and decoding at the receivers and substantially reducing the decoding complexity. We utilize BD precoding \cite{Wiesel2008} for SD of the effective MIMO channels of the PMs owing to its low-complexity and enhanced performance for critically-loaded systems \cite{Spencer2004}, and higher-order generalized singular value decomposition (HO-GSVD)-based precoding \cite{Ponnapalli2011}, \cite{VanLoan2015} for SD of the effective MIMO channel for the CM. Furthermore, we allow the exclusion of a subset of users from decoding the CM. Moreover, we formulate a non-convex WSR optimization problem and solve it via successive convex approximation (SCA) \cite{Razaviyayn2014} to obtain a local optimum. Lastly, we utilize computer simulations to characterize the performance of the proposed SD MIMO-RSMA. The main contributions of this paper can be summarized as follows. 
\begin{itemize}
    \item We propose a novel SD MIMO-RSMA scheme which utilizes HO-GSVD \cite{Ponnapalli2011, VanLoan2015} and BD \cite{Spencer2004} for SD of the effective MIMO channels for the CM and PMs, respectively. The proposed scheme also allows the exclusion of a subset of the users from decoding the CM.
    \item We formulate a non-convex WSR optimization problem and solve it via SCA \cite{Razaviyayn2014}.
    \item Lastly, through computer simulations, we compare the performance of the proposed SD MIMO-RSMA with the performances of MIMO-RSMA in \cite{Flores2019}, linear BD precoding \cite{Spencer2004}, and the DPC UB \cite{Vishwanath2003}. Our simulation results reveal that, for both perfect and imperfect CSI, the proposed scheme outperforms MIMO-RSMA in \cite{Flores2019} and linear BD precoding.
\end{itemize}

The remainder of this paper is organized as follows. In Section \ref{sec:prelim}, we present the system model and a brief description of HO-GSVD. The proposed SD MIMO-RSMA precoding and decoding schemes are described in Section \ref{section:proposed}. The WSR maximization problem and a solution based on SCA are provided in Section \ref{section:wsr}. Simulation results are presented in Section \ref{section:sims}, and the paper is concluded in Section \ref{section:conc}.

\emph{Notation:} Boldface capital letters $\boldsymbol{X}$ and boldface lowercase letters $\boldsymbol{x}$ denote matrices and vectors, respectively. $\boldsymbol{X}^\mathrm{H}$, $\boldsymbol{X}^{+}$, and tr($\boldsymbol{X}$) denote the Hermitian transpose, pseudo-inverse, and trace of matrix $\boldsymbol{X}$, respectively. $\mathbb{C}^{m \times n}$ and $\mathbb{R}^{m \times n}$ denote the set of $m \times n$ matrices with complex-valued and real-valued entries, respectively. The $(i,j)$-th entry of matrix $\boldsymbol{X}$ is denoted by $(\boldsymbol{X})_{i,j}$ and the $i$-th element of vector $\boldsymbol{x}$ is denoted by $(\boldsymbol{x})_{i}$. $\boldsymbol{I}_{n}$ denotes the $n \times n$ identity matrix. The circularly symmetric complex-valued Gaussian distribution with mean $\boldsymbol{\mu}$ and covariance matrix $\boldsymbol{\Sigma}$ is denoted by $\mathcal{CN}(\boldsymbol{\mu},\boldsymbol{\Sigma})$; $\sim$ stands for ``distributed as". $\mathrm{E}\mkern-\thinmuskip\left[\cdot\right]$ denotes statistical expectation.

\section{Preliminaries}
\label{sec:prelim}
In this section, we present the system model and briefly describe the HO-GSVD, which is utilized for SD of the effective MIMO channel for the CM.

\subsection{System Model}
\label{section:sys_model}
We consider a downlink MIMO system comprising $K$ users equipped with $M_k, k = 1,\dots, K$, antennas and a BS equipped with $N$ antennas. We study the critically-loaded regime, i.e., $\sum_{k=1}^{K} M_{k} = N.$ 
The MIMO channel matrix between the BS and user $k$ is modeled as $\frac{1}{\sqrt{L_{k}}}\boldsymbol{H}_{k} \in \mathbb{C}^{M_{k} \times N},$ $k=1,\dots,K,$ where $L_{k}$ models the free-space path loss for user $k$, and the elements of matrix $\boldsymbol{H}_{k}$ model the small-scale fading effects. 

Let $\textit{K} = \{1,\dots,K\}.$ Furthermore, let $\textit{K}_{\mathrm{c}} = \{b_{1},\dots,b_{L}\},$ $L \leq K,$ denote the set of users who decode the CM. Hence, users $\textit{K} \setminus \textit{K}_{\mathrm{c}}$ do not decode the CM. Furthermore, let $\textrm{W}_{k}$ denote the message intended for user $k.$ We construct the CM and the private MIMO symbol vectors as follows \cite{Mao2018}.

For users $k \in \textit{K}_{\mathrm{c}},$ message $\textrm{W}_{k}$ is split into two parts, $\textrm{W}_{k} = \{\textrm{W}_{k}^{\mathrm{c}},\textrm{W}_{k}^{\mathrm{p}}\}$. Here, $\textrm{W}_{k}$ is split such that user $k$ is assigned a fraction $v_k, 0 \leq v_k \leq 1, \sum_{k \in \textit{K}_{\mathrm{c}}} v_k = 1,$ of the available bits in the CM. On the other hand, for users $k \not\in \textit{K}_{\mathrm{c}},$ we set $v_k = 0.$ Next, message parts $\textrm{W}_{k}^{\mathrm{c}},k \in \textit{K}_{\mathrm{c}},$ are jointly encoded into the CM symbol vector, $\boldsymbol{s}_{\mathrm{c}} \in \mathbb{C}^{M \times 1},$ $M = \textrm{min}(M_{b_1},\dots, M_{b_L}).$ 

Furthermore, the PM of user $k,$ i.e., $\textrm{W}_{k}^{\mathrm{p}}$ if $k \in \textit{K}_{\mathrm{c}}$ or $\textrm{W}_{k}$ if $k \not\in \textit{K}_{\mathrm{c}},$ is encoded into the private symbol vector $\boldsymbol{s}_{k} \in \mathbb{C}^{M_k \times 1}, k=1,\dots,K.$ Here, we assume that $\mathrm{E}\mkern-\thinmuskip\left[\boldsymbol{s}_{\mathrm{c}}\boldsymbol{s}_\mathrm{c}^\mathrm{H}\right] = \boldsymbol{I}_{M}$ and $\mathrm{E}\mkern-\thinmuskip\left[\boldsymbol{s}_{k}\boldsymbol{s}_{k}^\mathrm{H}\right] = \boldsymbol{I}_{M_{k}}, k=1,\dots,K.$

The CM and private symbol vectors, $\boldsymbol{s}_{\mathrm{c}}$ and $\boldsymbol{s}_{k}, k=1,\dots,K,$ are precoded using matrices $\boldsymbol{P}_{\mathrm{c}} \in \mathbb{C}^{N \times M}$ and $\boldsymbol{P}_{k} \in \mathbb{C}^{N \times M_{k}}, k=1,\dots,K,$ respectively. The transmit power constraint at the BS is given by
\begin{equation}
\label{eqn:pow_cos}
\textrm{tr}(\boldsymbol{P}_{\mathrm{c}}\boldsymbol{P}_{\mathrm{c}}^{\mathrm{H}}) + \sum_{k=1}^{K}\textrm{tr}(\boldsymbol{P}_{k}\boldsymbol{P}_{k}^{\mathrm{H}}) \leq P_{\mathrm{T}},
\end{equation}
where $P_{\mathrm{T}}$ is the available transmit power, and the received signal at user $k$ is given by
\begin{equation}
\label{eqn:abstract_receive}
\boldsymbol{y}_{k} = \frac{1}{\sqrt{L_{k}}}\boldsymbol{H}_{k}(\boldsymbol{P}_{\mathrm{c}}\boldsymbol{s}_{\mathrm{c}} + \sum_{k'=1}^{K}\boldsymbol{P}_{k'}\boldsymbol{s}_{k'}) + \boldsymbol{z}_{k},
\end{equation}
where $\boldsymbol{z}_{k} \sim \mathcal{CN}(\boldsymbol{0},\sigma^{2}\boldsymbol{I}_{M_{k}}) \in \mathbb{C}^{M_{k} \times 1}$ is the complex additive white Gaussian noise (AWGN) vector at user $k.$
 
\begin{remark}
In this paper, we assume that the users know their own MIMO channel matrices perfectly. On the other hand, in Section \ref{section:sims}, we evaluate the performance of the proposed scheme for both prefect and imperfect knowledge of the MIMO channel matrices of the users at the BS.
\end{remark}

\subsection{Higher-Order Generalized Singular Value Decomposition}
HO-GSVD \cite{Ponnapalli2011}, \cite{VanLoan2015} is a matrix decomposition that can simultaneously diagonalize an arbitrary number of matrices having the same number of columns. The decomposition is given in the following theorem.

\begin{theorem}
\label{th:hogsvd}
Let $\boldsymbol{A}_i \in \mathbb{C}^{m_i\times n},i=1,\dots,S,$ be $S$ matrices having full column rank, and $m_i \geq n, i=1,\dots,S.$ Then, there exist left-invertible matrices $\boldsymbol{U}_i \in \mathbb{C}^{m_{i} \times n}, i=1,\dots,S,$ invertible matrix $\boldsymbol{V} \in \mathbb{C}^{n \times n},$ and diagonal matrices $\boldsymbol{\Sigma}_i \in \mathbb{R}^{n \times n}$ such that:
\begin{equation}
\boldsymbol{A}_{i} = \boldsymbol{U}_{i}\boldsymbol{\Sigma}_{i}\boldsymbol{V}^{\mathrm{H}}, \label{eqn:HO_basis}
\end{equation}
where the unit-norm columns of $\boldsymbol{U}_i,i=1,\dots,S,$ and $\boldsymbol{V}$ contain the left and right higher-order singular vectors of $\boldsymbol{A}_{i},i=1,\dots,S,$ and matrices $\boldsymbol{\Sigma}_{i},i=1,\dots,S,$ contain their $n$ higher-order generalized singular values (HO-GSVs) on the main diagonal.
\end{theorem}
\begin{proof}
The proof directly follows from \cite{Ponnapalli2011}. See also \cite{VanLoan2015}.
\end{proof}

\begin{remark}
	The HO-GSVD reduces to the GSVD \cite{VanLoan1976} for $K=2$ and the singular value decomposition for $K=1,$ see \cite{Ponnapalli2011,VanLoan2015} for details.
\end{remark}

\section{Proposed SD MIMO-RSMA Precoding and Decoding Schemes}
\label{section:proposed}
In this section, we present the proposed SD MIMO-RSMA precoding and decoding schemes which simultaneously diagonalize the effective MIMO channel matrices of the users for both the CM and the PMs. Additionally, we derive  expressions for the resulting achievable user rates. 

\subsection{Precoder Design}
First, we consider the precoding for the CM. As the CM has to be decoded by all users $k \in \textit{K}_{\mathrm{c}},$ the precoded signals must lie in the intersection of the row space of the MIMO channel matrices of the users, i.e., $\operatorname{col}(\boldsymbol{H}_{b_1}^\mathrm{H}) \cap \dots \cap \operatorname{col}(\boldsymbol{H}_{b_L}^\mathrm{H}).$ To this end, we define $\boldsymbol{G}_{\mathrm{c}} \in \mathbb{C}^{N\times M},$ a matrix whose columns contain the first $M$ right singular vectors of matrix $\boldsymbol{H}_{\mathrm{c}} = 
	\begin{pmatrix}
		\boldsymbol{H}_{b_{1}}^\mathrm{T} &
		\dots &
		\boldsymbol{H}_{b_{L}}^\mathrm{T} 
	\end{pmatrix}^\mathrm{T}.$ 
As desired, the columns of $\boldsymbol{G}_{\mathrm{c}}$ lie in $\operatorname{col}(\boldsymbol{H}_{b_1}^\mathrm{H}) \cap \dots \cap \operatorname{col}(\boldsymbol{H}_{b_L}^\mathrm{H}).$ Next, in order to simultaneously diagonalize the effective MIMO channel matrices of the users, we utilize the HO-GSVD, given in Theorem \ref{th:hogsvd}, of matrices $\{\boldsymbol{H}_k\boldsymbol{G}_{\mathrm{c}}, k \in \textit{K}_{\mathrm{c}}\}$ to obtain:
\begin{equation}
	\label{eqn:HO_GSVD_no_orth}
	\boldsymbol{E}_{k}^+(\boldsymbol{H}_{k}\boldsymbol{G}_{\mathrm{c}})\boldsymbol{V}_{\mathrm{c}}^{-\mathrm{H}} = \boldsymbol{D}_{k},
\end{equation}
for $k \in \textit{K}_{\mathrm{c}},$ where $\boldsymbol{D}_{k} \in \mathbb{R}^{M \times M}$ is a diagonal matrix containing the higher-order generalized singular values of $\{\boldsymbol{H}_k\boldsymbol{G}_{\mathrm{c}}, k \in \textit{K}_{\mathrm{c}}\}$ on the main diagonal, and the columns of matrices $\boldsymbol{E}_{k} \in \mathbb{C}^{M_k \times M}$ and $\boldsymbol{V}_{\mathrm{c}} \in \mathbb{C}^{M \times M}$ are the corresponding higher-order left and right singular vectors, respectively. Here, $\boldsymbol{E}_{k}^+ \in \mathbb{C}^{M \times M_k}$ denotes the left inverse of $\boldsymbol{E}_{k}, k=1,\dots,K.$

Now, the precoding matrix for the CM is chosen as follows:
\begin{equation}
\label{eqn:Pc_no_orth}
\boldsymbol{P}_{\mathrm{c}} = \boldsymbol{G}_{\mathrm{c}}\boldsymbol{V}_{\mathrm{c}}^{-{\mathrm{H}}}\boldsymbol{\Delta}_{\mathrm{c}}^{\frac{1}{2}},
\end{equation}
where $\boldsymbol{\Delta}_{\mathrm{c}} = \textrm{diag}(p_{\mathrm{c},1},\dots, p_{\mathrm{c},M}) \in \mathbb{R}^{M \times M}$ is the diagonal power loading matrix. Furthermore, based on (\ref{eqn:HO_GSVD_no_orth}), the detection matrix for the CM at user $k \in \textit{K}_{\mathrm{c}}$ is chosen as $\boldsymbol{E}_{k}^+.$

On the other hand, for the PMs, the precoder matrices are designed based on conventional BD \cite{Spencer2004}, as described in the following. In accordance with \cite{Spencer2004}, for user $k,$ let $\boldsymbol{N}_{k} \in \mathbb{C}^{N \times M_{k}}$ denote the null space of matrix $\boldsymbol{\hat{H}}_{k} = \begin{pmatrix}
	\boldsymbol{{H}}_{1}^\mathrm{T} &
	\dots &
	\boldsymbol{{H}}_{k-1}^\mathrm{T} &
	\boldsymbol{{H}}_{k+1}^\mathrm{T} &
	\dots &
	\boldsymbol{{H}}_{K}^\mathrm{T}
\end{pmatrix}^\mathrm{T}.$ 
Then, the precoding matrix for the PM of user $k$ is chosen as follows:
\begin{equation}
	\label{eqn:priv_precoder}
	\boldsymbol{P}_{k} = \boldsymbol{N}_{k}\boldsymbol{V}_{k}\boldsymbol{\Delta}_{k}^{\frac{1}{2}}, 
\end{equation}
and the detection matrix is chosen as $\boldsymbol{U}_{k}^\mathrm{H}.$ Here, $\boldsymbol{U}_{k},\boldsymbol{V}_{k} \in \mathbb{C}^{M_{k} \times M_{k}}$ are unitary matrices containing the left and right singular vectors of $\boldsymbol{H}_{k}\boldsymbol{N}_{k},$ and $\boldsymbol{\Delta}_k = \textrm{diag}(p_{k,1} ,\dots, p_{k,M_k}) \in \mathbb{R}^{M_k \times M_k}$ is the diagonal power loading matrix. Now, based on the above precoders, the power constraint in (\ref{eqn:pow_cos}) can be simplified as follows:
\begin{equation}
\label{eqn:updated_power_constraint}
\mathrm{tr}({\boldsymbol{V}_{\mathrm{c}}^{\mathrm{-H}}\boldsymbol{\Delta}_\mathrm{c}\boldsymbol{V}_{\mathrm{c}}^\mathrm{-1}}) + \sum_{k = 1}^{K} \sum_{l=1}^{M_{k}} p_{k,l} \leq P_\mathrm{T},
\end{equation}
and the received signal for user $k$ in (\ref{eqn:abstract_receive}) can be simplified to
\begin{equation}
\label{eqn:receive_alt}
\boldsymbol{y}_{k} = \frac{1}{\sqrt{L_{k}}}\boldsymbol{H}_{k}(\boldsymbol{P}_{\mathrm{c}}\boldsymbol{s}_{\mathrm{c}} + \boldsymbol{P}_{k}\boldsymbol{s}_{k}) + \boldsymbol{z}_{k}.
\end{equation}

\begin{remark}
	We note that, in (\ref{eqn:receive_alt}), the inter-user inference (IUI) from the others PMs is completely eliminated due to BD \cite{Spencer2004}.
\end{remark}

\subsection{Decoding Scheme}
In this section, we describe the proposed two-stage decoding scheme. For users $k \in \textit{K}_\mathrm{c},$ first, the CM is decoded. Next, the interference from the decoded CM symbols is eliminated via SIC and the PM is decoded. On the other hand, for $k \not\in \textit{K}_\mathrm{c},$ the PM is directly decoded treating the CM as noise.

\subsubsection{Decoding the Common Message}
For user $k \in \textit{K}_\mathrm{c},$ as mentioned earlier, $\boldsymbol{E}_{k}^+$ is utilized as the detection matrix to obtain the signal
\begin{equation}
\label{eqn:decode_comm}
\scalebox{0.95}{\mbox{\ensuremath{\displaystyle \boldsymbol{\Tilde{y}}_{k} = \boldsymbol{E}_{k}^+\boldsymbol{y}_{k} = \frac{1}{\sqrt{L_{k}}}(\boldsymbol{D}_{k}\boldsymbol{\Delta}_{\mathrm{c}}^{\frac{1}{2}}\boldsymbol{s}_{\mathrm{c}} + \boldsymbol{W}_{k}\boldsymbol{\Delta}_{k}^{\frac{1}{2}}\boldsymbol{s}_{k}) + \boldsymbol{\Tilde{z}}_{k}}}},
\end{equation}
where $\boldsymbol{\Tilde{z}}_{k} \sim \mathcal{CN}(\boldsymbol{0},\sigma^{2}\boldsymbol{E}_{k}^+(\boldsymbol{E}_{k}^+)^\mathrm{H}) \in \mathbb{C}^{M \times 1}$ is the processed AWGN noise vector and $\boldsymbol{W}_{k} = \boldsymbol{E}_{k}^+(\boldsymbol{H}_{k}\boldsymbol{N}_{k})\boldsymbol{V}_{k}$ is the matrix capturing the interference from the PM. 

\begin{remark}
	We note that, in (\ref{eqn:decode_comm}), due to the HO-GSVD, for all users, the effective MIMO channels for the CM are simultaneously diagonalized.
\end{remark}

Next, the $l$-th element of the common symbol vector $(\boldsymbol{s}_\mathrm{c})_{l}, l=1,\dots,M,$ is decoded element-by-element based on $(\boldsymbol{\Tilde{y}}_{k})_{l},$ i.e.,
\begin{align}
\label{eqn:comm_ele_ele}
\scalebox{0.95}{\mbox{\ensuremath{\displaystyle (\boldsymbol{\Tilde{y}}_{k})_{l}}}} &\scalebox{0.95}{\mbox{\ensuremath{\displaystyle {}= \frac{1}{\sqrt{L_{k}}}(\boldsymbol{D}_{k})_{l,l}\sqrt{(\boldsymbol{\Delta}_{\mathrm{c}})_{l,l}}(\boldsymbol{s}_{\mathrm{c}})_{l}}}} \nonumber\\&\quad\scalebox{0.95}{\mbox{\ensuremath{\displaystyle {}+ \frac{1}{\sqrt{L_{k}}} \sum_{i=1}^{M_{k}}(\boldsymbol{W}_{k})_{l,i}\sqrt{(\boldsymbol{\Delta}_{k})_{i,i}}(\boldsymbol{s}_{k})_{i} + (\boldsymbol{\tilde{z}}_{k})_{l}}}},
\end{align}
treating the interference from the PM as noise. Next, after successful decoding of the CM, SIC is performed to eliminate the interference from the received signal, $\boldsymbol{y}_{k},$ to obtain the interference-free signal
\begin{equation}
	\label{eqn:priv_SIC}    
	\scalebox{0.95}{\mbox{\ensuremath{\displaystyle \boldsymbol{{y}}'_{k} = \boldsymbol{y}_{k} - \frac{1}{\sqrt{L_{k}}}\boldsymbol{H}_{k}\boldsymbol{P}_{\mathrm{c}}\boldsymbol{s}_\mathrm{c} = \frac{1}{\sqrt{L_{k}}}\boldsymbol{H}_{k}\boldsymbol{P}_{k}\boldsymbol{s}_{k} + \boldsymbol{z}_{k}}}}.
\end{equation}

\begin{remark}
	We note that the SIC shown in (\ref{eqn:priv_SIC}) can also be performed element-by-element.
\end{remark}

\subsubsection{Decoding the Private Messages}
First, we describe the procedure for user $k \in \textit{K}_{\mathrm{c}}.$ Based on $\boldsymbol{{y}}'_{k}$ in (\ref{eqn:priv_SIC}), and utilizing $\boldsymbol{U}_{k}^{\mathrm{H}}$ as the detection matrix, we obtain signal
\begin{equation}
\label{eqn:det_priv_SIC}
\scalebox{0.95}{\mbox{\ensuremath{\displaystyle \boldsymbol{{y}}''_{k} = \boldsymbol{U}_{k}^{\mathrm{H}}\boldsymbol{y}'_{k} =  \frac{1}{\sqrt{L_{k}}}\boldsymbol{\Sigma}_{k}\boldsymbol{\Delta}_{k}^{\frac{1}{2}}\boldsymbol{s}_{k} + \boldsymbol{z}''_{k}}}},
\end{equation}
where $\boldsymbol{z}''_{k} \sim \mathcal{CN}(\boldsymbol{0},\sigma^{2}\boldsymbol{I}_{M_{k}}) \in \mathbb{C}^{M_{k} \times 1}$ because $\boldsymbol{U}_{k}^{\mathrm{H}}, k=1,\dots,K,$ is unitary.

\begin{remark}
	We note that, due to BD \cite{Spencer2004}, the effective MIMO channel of user $k$ for the PM is simultaneously diagonalized. 
\end{remark}

Next, as for the CM, the elements of symbol vector $(\boldsymbol{s}_{k})_{l}, l=1,\dots,M_k,$ are decoded element-by-element based on the $l$-th element of $\boldsymbol{{y}}''_{k},$ i.e.,
\begin{equation}
\label{eqn:priv_kcp_ele_ele}
(\boldsymbol{{y}}''_{k})_{l} = \frac{1}{\sqrt{L_{k}}}(\boldsymbol{\Sigma}_{k})_{l,l}\sqrt{(\boldsymbol{\Delta}_{k})_{l,l}}(\boldsymbol{s}_{k})_{l} + (\boldsymbol{z}''_{k})_{l},
\end{equation}
for $l = 1,\dots, M_{k}$.

On the other hand, for users $k \not\in \textit{K}_{\mathrm{c}},$ the PM is directly decoded treating the interference from the CM as noise. To this end, as earlier, $\boldsymbol{U}_{k}^{\mathrm{H}}$ is utilized as the detection matrix to obtain the processed signal
\begin{equation}
\label{eqn:Alt_1_priv}    
\boldsymbol{{y}}''_{k} = \boldsymbol{U}_{k}^{\mathrm{H}}\boldsymbol{y}_{k} = \frac{1}{\sqrt{L_{k}}}(\boldsymbol{\Sigma}_{k}\boldsymbol{\Delta}_{k}^{\frac{1}{2}}\boldsymbol{s}_{k} + \boldsymbol{W}_{\mathrm{c},k}\boldsymbol{\Delta}_{\mathrm{c}}^{\frac{1}{2}}\boldsymbol{s}_{\mathrm{c}}) + \boldsymbol{z}''_{k}, 
\end{equation}
where $\boldsymbol{W}_{\mathrm{c},k} = \boldsymbol{U}_{k}^{\mathrm{H}}(\boldsymbol{H}_{k}\boldsymbol{G}_{\mathrm{c}})\boldsymbol{V}_{\mathrm{c}}^{-\mathrm{H}} \in \mathbb{C}^{M_{k} \times M}$ is the interference caused by the CM to the PM.

Next, treating the interference from the CM as noise, the $l$-th element of the symbol vector $(\boldsymbol{s}_{k})_{l}$ is decoded based on the $l$-th element of $\boldsymbol{{y}}''_{k},$ i.e., 
\begin{align}
\label{eqn:alt_1_ele_ele}
\scalebox{0.95}{\mbox{\ensuremath{\displaystyle (\boldsymbol{{y}}''_{k})_{l}}}} &= \scalebox{0.95}{\mbox{\ensuremath{\displaystyle \frac{1}{\sqrt{L_{k}}}(\boldsymbol{\Sigma}_{k})_{l,l}\sqrt{(\boldsymbol{\Delta}_{k})_{l,l}}(\boldsymbol{s}_{k})_{l}}}} \nonumber\\&\quad\scalebox{0.95}{\mbox{\ensuremath{\displaystyle {}+ \frac{1}{\sqrt{L_{k}}}\sum_{i=1}^{M}(\boldsymbol{W}_{\mathrm{c},k})_{l,i}\sqrt{(\boldsymbol{\Delta}_{\mathrm{c}})_{i,i}}(\boldsymbol{s}_{\mathrm{c}})_{i}+ (\boldsymbol{z}''_{k})_{l}}}},
\end{align}
for $l = 1,\dots, M_{k}$.

\begin{remark}
	We note that, unlike for users $k \not\in \textit{K}_{\mathrm{c}},$ for users $k \in \textit{K}_{\mathrm{c}},$ decoding and SIC of the CM must be carried out even if no CM bits are allocated for user $k,$ i.e., $v_k = 0.$ 
\end{remark}

\begin{remark}
	For $k \not\in \textit{K}_{\mathrm{c}},$ the proposed precoding and decoding schemes are purely linear. Hence, the proposed SD MIMO-RSMA is a hybrid scheme combining non-linear and linear decoding.
\end{remark}

\begin{remark}
	In the following, in order to simplify our notation, we set $\boldsymbol{W}_{\mathrm{c},k} = \boldsymbol{0}$ for users $k \in \textit{K}_{\mathrm{c}}.$
\end{remark}

\subsection{Achievable Rates}
\label{sec:ar}
For the CM and users $k \in \textit{K}_{\mathrm{c}},$ based on (\ref{eqn:comm_ele_ele}), the achievable rate of the $l$-th element of the common symbol vector, $(\boldsymbol{s}_{\mathrm{c}})_{l}$, $l = 1,\dots, M,$ is given by
\begin{equation}
\label{eqn:rate_common}
    \scalebox{0.9}{\mbox{\ensuremath{\displaystyle R_{\mathrm{c},l}^{k} = \log_{2}\bigg{(}1+\frac{\frac{1}{L_{k}}(\boldsymbol{D}_{k})_{l,l}^{2}p_{\mathrm{c},l}}{\frac{1}{L_{k}}\sum_{i=1}^{M_{k}}|(\boldsymbol{W}_{k})_{l,i}|^{2}p_{k,i} + \sigma^{2}(\boldsymbol{E}_{k}^+(\boldsymbol{E}_{k}^+)^\mathrm{H})_{l,l}}\bigg{)}}}}.
\end{equation}
Now, as the CM must be decoded by all users $k \in \textit{K}_{\mathrm{c}},$ the rate of the $l$-th element is chosen as
\begin{equation}
\label{eqn:comm_min}
 R_{\mathrm{c},l} = \min_{k \in \textit{K}_\mathrm{c}}(R^{k}_{\mathrm{c},l}).
\end{equation}

Next, for the PMs, based on (\ref{eqn:priv_kcp_ele_ele}) and (\ref{eqn:alt_1_ele_ele}), the achievable rate of the $l$-th element of the private symbol vector $(\boldsymbol{s}_{k})_{l}, l=1,\dots,M_k,$ of user $k,$ is given by
\begin{equation}
\label{eqn:generic_priv}
    R_{k,l} = \log_{2}\bigg{(}1 + \frac{\frac{1}{L_{k}}(\boldsymbol{\Sigma}_{k})_{l,l}^{2} p_{k,l}}{ \frac{1}{L_{k}}\sum_{i=1}^{M}|(\boldsymbol{W}_{\mathrm{c},k})_{l,i}|^{2}p_{\mathrm{c},i} + \sigma^{2}}\bigg{)},
\end{equation}
where $\boldsymbol{W}_{\mathrm{c},k} = \boldsymbol{0}$ for $k \in \textit{K}_{\mathrm{c}}.$

\subsection{Computational Complexity}
For BD and HO-GSVD, computing the precoding and detection matrices for the CM and the PM entails a complexity of $\operatorname{\mathcal{O}}(N^3)$ per user \cite{Spencer2004,Ponnapalli2011,VanLoan2015}. Hence, the proposed SD MIMO-RSMA scheme entails a total complexity of $\operatorname{\mathcal{O}}(K N^3).$

\section{WSR Maximization}
\label{section:wsr}
In this section, we formulate the WSR maximization problem and present a locally-optimal solution via SCA \cite{Razaviyayn2014}. 

\subsection{Weighted Sum Rate}
Based on the achievable rates in Section \ref{sec:ar}, the WSR is given as follows:
\begin{equation}
\label{eqn:WSR_def}
R_{\mathrm{wsr}}= \sum_{k=1}^{K}\sum_{l=1}^{M}w_{k}v_{k}{R}_{\mathrm{c},l} + \sum_{k=1}^{K}\sum_{l=1}^{M_{k}}w_{k}{R}_{k,l},
\end{equation}
where $w_k,k=1,\dots,K,$ $0 \leq w_{k} \leq 1,$ and $\sum_{k=1}^{K}w_{k} = 1,$ denote fixed weights, which can be chosen to adjust the rates of the users during power allocation \cite[Sec. 4]{WangGiannakis2011}, and $v_k, k=1,\dots,K,$ as described in Section \ref{section:sys_model}, denotes the fraction of the CM assigned to user $k.$

Now, the WSR maximization problem can be formulated as follows:
\begin{maxi!}
		{\substack{{p}_{\mathrm{c},l} \geq {0}, \forall\,l \\ {p}_{k,l} \geq {0}, \forall\,k, \forall\,l}}
		{R_\mathrm{wsr}}
		{\label{opt:wsr_noncov}}
		{\mathrlap{R_\mathrm{wsr}^\star =}\nonumber\\}
		\addConstraint{\text{C1: } \scalebox{0.85}{\mbox{\ensuremath{\displaystyle \mathrm{tr}({\boldsymbol{V}_{\mathrm{c}}^{\mathrm{-H}}\boldsymbol{\Delta}_\mathrm{c}\boldsymbol{V}_{\mathrm{c}}^\mathrm{-1}}) + \sum_{k = 1}^{K} \sum_{l=1}^{M_{k}} p_{k,l} \leq P_\mathrm{T}}}}.\label{cons:t31}}{}{}
\end{maxi!}

Solving (\ref{opt:wsr_noncov}) optimally entails a high complexity as it is non-convex in the optimization variables $p_{\mathrm{c},l}, l = 1,\dots, M$, and $p_{k,l}, l = 1,\dots, M_{k}, k=1,\dots,K.$ Hence, in the following, we solve (\ref{opt:wsr_noncov}) via SCA \cite{Razaviyayn2014} to obtain a low-complexity locally-optimal solution.

\subsection{Successive Convex Approximation}
In SCA, first, an inner convex problem is formulated by replacing the non-convex parts of $R_\mathrm{wsr}$ in (\ref{opt:wsr_noncov}) by their first-order approximations. Next, the inner convex problem is solved iteratively until convergence, up to a numerical tolerance $\epsilon.$ In each iteration, the first-order approximations in the objective function are updated based on the results of the previous iteration, and the revised objective function is utilized. The procedure is described in detail below.

Let $\boldsymbol{p}_\mathrm{c} = \{p_{\mathrm{c},l}, l = 1,\dots, M\},$ $\boldsymbol{p}_k = \{p_{k,l}, l = 1,\dots, M_k\}, k=1,\dots,K,$ denote sets containing the optimization variables for the CM and the PMs, respectively. Furthermore, Let $\boldsymbol{p}_\mathrm{c}^{(n)}$ and $\boldsymbol{p}_k^{(n)}, k=1,\dots,K,$ denote the values of the sets in the $n$-th iteration, $n=1,2,\dots,$ with initial values $\boldsymbol{p}_\mathrm{c}^{(0)} = \boldsymbol{0}$ and $\boldsymbol{p}_k^{(0)} = \boldsymbol{0}, k=1,\dots,K.$

\begin{figure*}
	\begin{align}
		\scalebox{0.9}{\mbox{\ensuremath{\displaystyle \tilde{R}_{\mathrm{c},l}^{k,(n)}}}} &{}= \scalebox{0.8}{\mbox{\ensuremath{\displaystyle \textrm{log}_{2}\bigg{(} \frac{1}{{L_{k}}}\sum_{i=1}^{M_{k}}|(\boldsymbol{W}_{k})_{l,i}|^{2} p_{k,i} + \sigma^{2}(\boldsymbol{E}_{k}^+(\boldsymbol{E}_{k}^+)^\mathrm{H})_{l,l} + \frac{1}{{L_{k}}}(\boldsymbol{D}_{k})_{l,l}^{2}p_{\mathrm{c},l}\bigg{)}}}}
			\breakifsinglecolumn
			- \scalebox{0.9}{\mbox{\ensuremath{\displaystyle {\frac{ \frac{1}{{L_{k}}}\sum_{i=1}^{M_{k}}|(\boldsymbol{W}_{k})_{l,i}|^{2}p_{k,i}}{\textrm{log}(2)\big{(}\frac{1}{{L_{k}}}\sum_{j=1}^{M_{k}}|(\boldsymbol{W}_{k})_{l,j}|^{2}p_{k,j}^{(n-1)} + \sigma^{2}(\boldsymbol{E}_{k}^+(\boldsymbol{E}_k^+)^\mathrm{H})_{l,l}\big{)}}}}}}, \label{eqn:rclkn} \\
		\scalebox{0.9}{\mbox{\ensuremath{\displaystyle \tilde{R}_{k,l}^{(n)}}}} &{}= \scalebox{0.8}{\mbox{\ensuremath{\displaystyle \textrm{log}_{2}\bigg{(}\frac{1}{L_{k}}\sum_{i=1}^{M}|(\boldsymbol{W}_{\mathrm{c},k})_{l,i}|^{2} p_{\mathrm{c},i} + \sigma^{2} + \frac{1}{L_{k}}|\boldsymbol{(\Sigma}_{k})_{l,l}|^{2}p_{k,l}\bigg{)}}}}
			\breakifsinglecolumn
		- \scalebox{0.9}{\mbox{\ensuremath{\displaystyle {\frac{ \frac{1}{L_{k}}\sum_{i=1}^{M}|(\boldsymbol{W}_{\mathrm{c},k})_{l,i}|^{2} p_{\mathrm{c},i}}{\textrm{log}(2)\big{(}\frac{1}{L_{k}}\sum_{j=1}^{M}|(\boldsymbol{W}_{\mathrm{c},k})_{l,j}|^{2}p_{\mathrm{c},j}^{(n-1)} + \sigma^{2}\big{)}}}}}} \label{eqn:rkln}
	\end{align}
	\hrulefill
	\vspace{-0.5cm}
\end{figure*}

In iteration $n,$ $n=1,2,\dots,$ a convex approximation of $R_\mathrm{wsr}$ is constructed, as follows:
\begin{align}
	\tilde{R}_\mathrm{wsr}^{(n)} = \sum_{k=1}^{K}\sum_{l=1}^{M}w_{k}v_{k}\tilde{R}_{\mathrm{c},l}^{(n)} + \sum_{k=1}^{K}\sum_{l=1}^{M_{k}}w_{k}\tilde{R}_{k,l}^{(n)},
\end{align}
where $\tilde{R}_{\mathrm{c},l}^{(n)} = \min_{k \in \textit{K}_\mathrm{c}} \tilde{R}_{\mathrm{c},l}^{k,(n)},$ and the non-constant parts of $\tilde{R}_{\mathrm{c},l}^{k,(n)},$ $l=1,\dots,M, k=1,\dots,K,$ and $\tilde{R}_{k,l}^{(n)},$ $l=1,\dots,M_k, k=1,\dots,K.$ are shown on the top of the next page in (\ref{eqn:rclkn}) and (\ref{eqn:rkln}), respectively.

Next, an inner convex-optimization problem is constructed based on $\tilde{R}_\mathrm{wsr}^{(n)}$ as follows.
\begin{maxi!}
	{\substack{{p}_{\mathrm{c},l} \geq {0}, \forall\,l \\ {p}_{k,l} \geq {0}, \forall\,k, \forall\,l}}
	{\tilde{R}_\mathrm{wsr}^{(n)}}
	{\label{opt:wsr_conv}}
	{R_\mathrm{wsr}^{(n)\star} =}
	\addConstraint{\text{C1.}}{}{}
\end{maxi!}

Now, employing SCA \cite{Razaviyayn2014}, the inner convex optimization problem, (\ref{opt:wsr_conv}), is solved using standard convex optimization tools \cite{Boyd2004} to obtain the optimal value $R_\mathrm{wsr}^{(n)\star}$ and the corresponding optimal solution $\boldsymbol{p}_\mathrm{c}^{(n)}, \boldsymbol{p}_k^{(n)}, k=1,\dots,K.$ Then, the optimal solution is utilized to construct the first-order approximation $\tilde{R}_\mathrm{wsr}^{(n+1)}$ for the next iteration. This process is repeated, which gradually tightens the first-order approximation around a local optimum of (\ref{opt:wsr_noncov}). Hence, the obtained optimal values $R_\mathrm{wsr}^{(n)\star}, n=1,2,\dots,$ converge to a local optimum of (\ref{opt:wsr_noncov}) \cite{Razaviyayn2014}. In our case, the iterations are continued until convergence up to a numerical tolerance $\epsilon.$ The algorithm is summarized in Algorithm \ref{alg:1}.

\begin{figure}
\begin{algorithm}[H]
\small
\begin{algorithmic}[1]
\STATE {Initialize $R_\mathrm{wsr}^{(0)\star} = -\infty,$ $\boldsymbol{p}^{(0)}_\mathrm{c} = \boldsymbol{0}, \boldsymbol{p}^{(0)}_{k} = \boldsymbol{0}, k = 1,\dots, K$, iteration index $n = 0$, and numerical tolerance $\epsilon$}
\REPEAT
\STATE {$n\gets n+1$} 
\STATE {Solve (\ref{opt:wsr_conv}) to obtain optimal value $R_\mathrm{wsr}^{(n)\star}$ and the corresponding solution $\boldsymbol{p}^{(n)}_\mathrm{c},\boldsymbol{p}_{k}^{(n)},k = 1,\dots, K$}
\UNTIL {$|R_\mathrm{wsr}^{(n)\star} - R_\mathrm{wsr}^{(n-1)\star}| \leq \epsilon$}
\STATE {Return $\boldsymbol{p}_\mathrm{c}^{(n)}, \boldsymbol{p}_{k}^{(n)}$, $k=1,\dots, K,$ as the solution}
\end{algorithmic}
\caption{Algorithm for solving (\ref{opt:wsr_noncov}) via SCA.}
\label{alg:1}
\end{algorithm}
\vspace{-0.5cm}
\end{figure}

\section{Simulation Results}
\label{section:sims}
In this section, we evaluate the performance of the proposed SD MIMO-RSMA and compare it with that of MIMO-RSMA in \cite{Flores2019}, linear BD precoding \cite{Spencer2004}, and the DPC UB \cite{Vishwanath2003} for perfect and imperfect CSI. As mentioned in Section \ref{section:sys_model}, we consider a critically-loaded downlink MIMO system with $N=\sum_{k=1}^{K}M_{k}$. For our simulations, the number of users is $K=4,$ and the numbers of antennas at the BS and the users are $N = 16$ and $M_k = 4, k=1,\dots,K,$ respectively. The free-space path loss is set to $L_{k} = d_{k}^2$, where ${d}_{k}$ is the distance (in meters) between the BS and user $k,$ the noise variance is set to $\sigma^{2} = -35$ dBm, and the numerical tolerance $\epsilon$ in Algorithm \ref{alg:1} is $\epsilon = 10^{-6}$.

Furthermore, for the proposed scheme, we maximize the WSR over all $2^K-1$ possible combinations of user indices in set $\textit{K}_\mathrm{c}.$ That is, for each combination of user indices in $\textit{K}_\mathrm{c},$ we set $v_{k} = \frac{1}{|\textit{K}_\mathrm{c}|},k \in \textit{K}_\mathrm{c},$ otherwise $v_{k} = 0,$ and utilize Algorithm \ref{alg:1} to maximize the WSR. Then, we choose the combination leading to the maximum sum rate (SR). Additionally, we evaluate the SR for the proposed scheme without user exclusion, i.e., for $\textit{K}_\mathrm{c} = \textit{K}.$ On the other hand, for MIMO-RSMA in \cite{Flores2019}, linear BD precoding, and the DPC UB, we obtain the SR based on \cite{Flores2019}, \cite{Spencer2004}, and \cite{Vishwanath2003}, respectively. For all considered schemes, we average the SR over multiple MIMO channel realizations so as to obtain a $99\%$ confidence interval of $\pm 1$ BPCU.

\subsection{Channel Model}
In our simulations, we consider the case where the MIMO channel matrices of the users are correlated, e.g., when the users are located close to each other, which is most beneficial for RSMA because the correlations prevent low-complexity linear schemes such as BD \cite{Spencer2004} and ZF \cite{Wiesel2008} precoding from fully exploiting the spatial DoFs. For generating the MIMO channel matrices of the users, we define $\alpha \in \mathbb{R},$ $0 \leq \alpha \leq 1,$ to specify the correlation and matrices $\boldsymbol{G}_{k}, k = 1,\dots,K,$ whose elements are drawn from independent and identically distributed (i.i.d.) Gaussian distributions, i.e., $({\boldsymbol{G}_{k}})_{i,j} \sim \mathcal{CN}(0,1)$, $i=1,\dots,M_k, j=1,\dots,N, k=1,\dots,K,$ to model the small-scale fading effects. The MIMO channel matrices of the users are constructed as follows: $\boldsymbol{H}_1 = \boldsymbol{G}_1,$ $\boldsymbol{H}_2 = \boldsymbol{G}_2,$ $\boldsymbol{H}_3 = \alpha\boldsymbol{G}_1 + \sqrt{1-\alpha^2}\boldsymbol{G}_3,$ and $\boldsymbol{H}_4 = \alpha\boldsymbol{G}_2 + \sqrt{1-\alpha^2}\boldsymbol{G}_4,$ see also the channel model utilized in \cite{Mao2018}. We note that the elements of $\boldsymbol{H}_k$ are i.i.d. Gaussian distributed with zero mean and unit variance. However, $\mathrm{E}\mkern-\thinmuskip\left[(\boldsymbol{H}_1)_{i,j} (\boldsymbol{H}_3)_{i,j}\right] = \mathrm{E}\mkern-\thinmuskip\left[(\boldsymbol{H}_2)_{i,j} (\boldsymbol{H}_4)_{i,j}\right] = \alpha, i=1,\dots,M_1, j=1,\dots,N.$ Uncorrelated channel matrices are obtained by setting $\alpha=0.$ Next, for imperfect CSI, we model the estimated MIMO channel matrices of user $k$ at the BS as follows:
\begin{align}
	\tilde{\boldsymbol{H}}_k = \boldsymbol{H}_{k} + \Delta\boldsymbol{H}_{k}, \label{eqn:imodel}
\end{align}
where the elements of $\Delta\boldsymbol{H}_{k} \in \mathbb{C}^{M_k\times N},k=1,\dots,K,$ are i.i.d. Gaussian distributed with zero mean and variance $\mu^2,$ i.e., $(\Delta\boldsymbol{H}_{k})_{i,j} \sim \mathcal{CN}(0,\mu^2), i=1,\dots,M_k,j=1,\dots,N.$ Moreover, for the proposed and the baseline schemes, to study the impact of imperfect CSI at the BS, we obtain the precoder and detection matrices based on $\tilde{\boldsymbol{H}}_k, k=1,\dots,K,$ but compute the SR based on $\boldsymbol{H}_{k}, k=1,\dots,K,$ taking into account the additional IUI introduced due to the imperfections.

\begin{figure}
     \centering
     \begin{subfigure}{0.45\textwidth}
           \centering
           \includegraphics[width=0.85\textwidth]{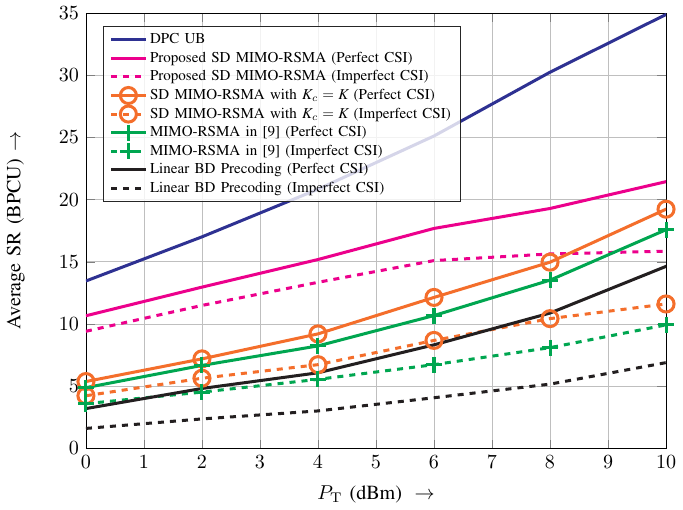}
           \caption{Correlated user MIMO channel matrices and $d_k=50 \text{ m},k=1,\dots,K.$}
            \label{fig:SR1.tikz}
     \end{subfigure}%
     \hspace{0.04\textwidth}%
     \begin{subfigure}{0.45\textwidth}
           \centering
           \includegraphics[width=0.85\textwidth]{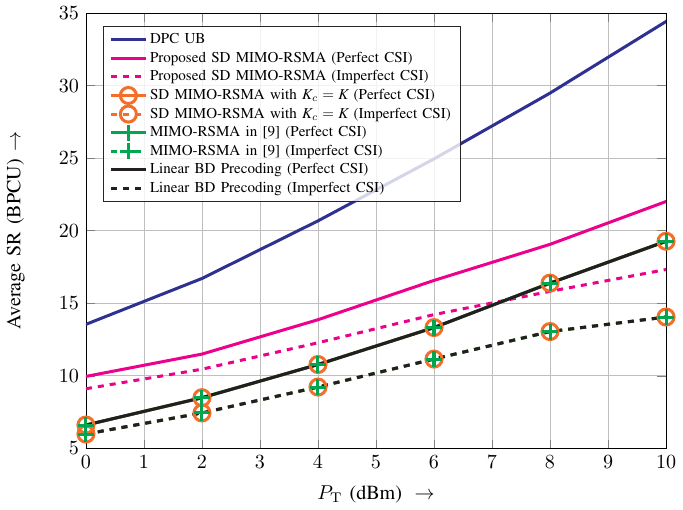}
           \caption{Uncorrelated user MIMO channel matrices and $d_k=\{250,250,50,50\} \text{ m}.$}
		   \label{fig:SR2.tikz}
    \end{subfigure}%
    \caption{Average SR vs ${P}_\mathrm{T}$ for $K=4, $ $M_k=4, k=1,\dots,K, N=16.$}
    \vspace{-0.5cm}
\end{figure}

\subsection{Results}
In Figure \ref{fig:SR1.tikz}, we show the average SR as a function of the transmit power $P_\mathrm{T}$ for the case where the MIMO channel matrices of the users are correlated. Here, the distances between the BS and the $K=4$ users are set to be equal, i.e., $d_k = 50 \text{ m},k=1,\dots,K.$ We consider the highly correlated case with $\alpha=0.8,$ and for imperfect CSI, we choose $\mu^2 = 0.1.$ From Figure \ref{fig:SR1.tikz}, we observe that, for both perfect and imperfect CSI, the proposed SD MIMO-RSMA benefits from user exclusion and significantly outperforms MIMO-RSMA in \cite{Flores2019} and linear BD precoding. This is because the proposed scheme is able to exploit the effective MIMO channel of the CM and user exclusion to enhance the achievable rate of the CM.

Next, in Figure \ref{fig:SR2.tikz}, we consider the same case as in Figure \ref{fig:SR1.tikz} but with uncorrelated MIMO channels for the users, i.e., $\alpha=0.$ Furthermore, we employ unequal user distances, i.e., $d_k = \{250, 250, 50, 50\}\text{ m}.$ From Figure \ref{fig:SR2.tikz}, we observe that, for both perfect and imperfect CSI, the proposed scheme \emph{without} user exclusion and MIMO-RSMA in \cite{Flores2019} have no performance advantage over linear BD precoding. This is because the achievable rate of the CM is negligible owing to the uncorrelated MIMO channels of the users and the two weak users with $d_k = 250 \text{ m}.$ Nevertheless, for perfect and imperfect CSI, the proposed SD MIMO-RSMA \emph{with} user exclusion outperforms both MIMO-RSMA in \cite{Flores2019} and linear BD precoding by excluding users with weak effective MIMO channels for the CM from decoding the CM, thereby enhancing the achievable rate of the CM.

\section{Conclusion}
\label{section:conc}
We considered the design of low-complexity precoding and decoding schemes for downlink MIMO-RSMA systems. To this end, to allow for element-by-element SIC and decoding at the users, we performed SD of the MIMO channel matrices of the users, for both the CM and the PMs. For the CM, SD was achieved via HO-GSVD \cite{Ponnapalli2011,VanLoan2015}, and for the PMs, BD \cite{Spencer2004} was utilized. Furthermore, in the proposed SD MIMO-RSMA scheme, users can be excluded from decoding the CM. We formulated a WSR maximization problem and found a locally-optimal solution via SCA \cite{Razaviyayn2014}. Our simulation results revealed that, for both perfect and imperfect CSI and the considered simulation scenarios with correlated and uncorrelated user MIMO channels, the proposed SD MIMO-RSMA outperformed MIMO-RSMA in \cite{Flores2019} and linear BD precoding.

In the proposed scheme, users $k \not\in \textit{K}_\mathrm{c}$ are served exclusively via the PM. Extensions of the proposed scheme where some users are served exclusively via the CM, i.e., have no PM, are also of interest for systems where some users have low rate requirement such as for \emph{Internet of Things} applications.

\bibliographystyle{IEEEtran}
\bibliography{IEEEabrv,references}

\end{document}